\documentclass[letterpaper, 10 pt, conference]{ieeeconf}
\usepackage{amssymb}
\usepackage{amsmath}
\usepackage{graphicx}

\let\ieeeproof\proof
\let\endieeeproof\endproof
\let\proof\relax
\let\endproof\relax

\usepackage{amsthm}

\let\proof\ieeeproof
\let\endproof\endieeeproof

\newtheorem{definition}{Definition}
\usepackage{subcaption}

\newtheorem{theorem}{Theorem}

\usepackage{dblfloatfix}
\usepackage{cite}
\usepackage{color}

\def\bs{\boldsymbol}
\def\mcal{\mathcal}
\def\mbb{\mathbb}
\def\mcx{\mcal{X}}
\def\mcy{\mcal{Y}}
\def\mcm{\mcal{M}}

\def\mcn{\mcal{N}}
\def\mcf{\mcal{F}}

\IEEEoverridecommandlockouts                   
\title{\LARGE \bf
Networked Multi-Resource Defense Capabilities in a General Lotto Game 
}

\author{Faezeh Shojaeighadikolaei, Keith Paarporn
\thanks{ F. Shojaeighadikolaei and K. Paarporn are with the Department of Computer Science, University of Colorado Colorado Springs. This work is supported by NSF grants \#2346791 and \#2541011. Contact: \texttt{ \{fshojaei,kpaarpor\}@uccs.edu}.
}
}

\begin{document}

\maketitle
\thispagestyle{empty}
\pagestyle{empty}

\begin{abstract}

Ensuring the security of complex systems involves the strategic allocation of defensive resources to prevent various types of attacks from succeeding. A defender often has multiple types of defensive assets at its disposal, where it must decide how to optimally deploy their heterogeneous capabilities across different attack types. In this paper, we formulate a multi-resource allocation problem in the form of a General Lotto game where a defender possesses various types of resources. A feature that we introduce is that their individual effectiveness against different types of attacks is characterized by a network weight matrix. In our analysis, we derive upper and lower bounds on the performance of the defender, and provide numerical evidence suggesting that they are tight. For the case of two attack types, we analytically prove that the bounds coincide, establishing an exact equilibrium characterization. We then numerically compare our proposed networked multi-resource architecture to an independent-defense benchmark from the existing literature. These results highlight fundamental and tractable structures underlying multi-attack-type defense problems.

\end{abstract}

\section{Introduction}
A common problem faced in security is how to distribute limited defensive resources against several types of threats. Applications range from cybersecurity, to critical infrastructure protection, to military operations, where a defender needs to apportion resources to protect a system against multiple types of attacks while considering the actions of strategic attackers \cite{11164058}. These problems become particularly natural when the attacker also faces limited resources to allocate across different attack types.
Analyzing how to optimally allocate resources is therefore important for informing security policies \cite{Shao2021Optimal,Yan2024Game-theoretical,Min2018Defense,Yan2024Game,tambe2011security,Gorbachuk2022GAME,Wang2026A}.

A crucial aspect of security policies is how to effectively leverage multiple types of defensive assets with heterogeneous qualities and costs. For example, defensive resources may serve distinct preventive and reactive roles in protecting a system against cyber attacks \cite{Shojaeighadikolaei}. Similarly, how to deploy distributed detection components in tandem with controllers for cyber--physical systems leads to nuanced design considerations \cite{BEHNAM2026110809,Atkinson2023Resource}.
More broadly, resource allocation under uncertainty also arises in critical infrastructure and energy systems, including hydropower scheduling and renewable energy integration \cite{Khojaste2025}.
Additionally, firms and health care organizations must invest in cybersecurity resources, including technical tools and employee training, to address threats such as data breaches and ransomware attacks \cite{Bidoki2026}. In cybersecurity, defenders may face distinct attack types, including phishing and other social-engineering attacks, which motivate a range of defensive mechanisms \cite{Pritom2026}.

In this paper, we investigate a game-theoretic model of strategic resource allocation, where a defender is tasked with protecting a system against multiple types of attacks by deploying several types of security resources. Each defensive resource type's effectiveness against the different attack types is heterogeneous, and these relationships can be described according to a networked weights matrix.
We consider the scenario where the attacker has dedicated, specialized resource types to employ against each attack type.
Hence, the attacker's resources are specialized by attack type, whereas the defender's resources may be effective against multiple attack types through the networked effectiveness structure. We model this interaction as a General Lotto-type game which captures the elements of strategic resource allocation and weakest-link security interdependence.

Resource allocation games such as the Colonel Blotto and General Lotto game have been used extensively to study how adversaries should allocate resources when facing multiple targets \cite{kovenock2021generalizations}. Extensions have studied heterogeneous values across targets, asymmetric budgets, and non-zero sum variations of the game \cite{kovenock2021generalizations,Aghajan2024Extension,Sonin2024Blotto}. Several of these works have also considered extensions with networked battles \cite{Guan2020Colonel,Guan2018Colonel,Aghajan2023The,Erat2024Colonel}. However, to the best of our knowledge, none of these models allow for flexible routing of defender resources across attack types. Indeed, our study extends existing work that considers multiple resource types in which all of them are specialized~\cite{aghajan2023multiresource,paarporn2025allocationheterogeneousresourcesgeneral}. In particular, these specialized-resource models correspond to a diagonal effectiveness structure, whereas our formulation allows defensive resource types to be effective against multiple attack types.

Our contributions in this paper are as follows. First, we introduce a networked multi-resource weakest-link General Lotto model that simultaneously accounts for heterogeneous defensive resource capabilities and flexible routing across attack types; see Section~\ref{sec:model} for details. Second, we establish upper and lower bounds on the defender's min-max and max-min values, respectively, for this formulation; these results, along with supporting numerical evidence suggesting that the bounds are tight, are presented in Section~\ref{sec:result}. Third, we show that in the two-attack-type case these bounds admit a closed-form expression and coincide, thereby exactly characterizing the equilibrium value in this specific setting; this is presented in Section~\ref{sec:two-nodes}. We then present numerical comparisons between these solutions and previous work with specialized resource heterogeneity.


\section{System Model}
\label{sec:model}

\subsection{Classical Weakest-Link Resource Allocation Game}

Here, we first formulate a model we refer to as the \emph{classical weakest-link game}.
There are two strategic players: the defender $\mcx$ and the attacker $\mcy$.
They compete over a system with $n$ nodes labeled $\mathcal{N} = \{1, \ldots, n \}$.

Both players allocate resources across the nodes. Let $x=(x_1,\ldots,x_n)\in\mathbb{R}_+^n$ denote a defense allocation vector, where $x_i$ represents the defensive effort expended on node $i$. Similarly, let $y=(y_1,\ldots,y_n)\in\mathbb{R}_+^n$ denote an attack allocation vector, where $y_i$ represents the attack effort directed toward node $i$. 
Given allocation vectors $x,y \in \mathbb{R}_+^n$, the payoff to the defender is defined as
\begin{equation}
\pi_\mcx(x,y)=
\mathbf{1}\{x_i \ge y_i \ \forall i \in \mcn \}.
\end{equation}
In words, the defender successfully defends the system if and only if it has a higher allocation than the attacker \emph{on every node}. This is known as the \emph{weakest-link objective}. Otherwise, the attacker successfully compromises the system. The attacker's payoff is given by $\pi_Y(x,y)=1-\pi_X(x,y)$.

In this formulation, the players are able to randomize their allocations in a way such that their total expected allocations satisfy their fixed resource budgets, $X,Y > 0$.
That is, the defender chooses a distribution $F_\mcx$ over $\mathbb{R}_+^n$ belonging to the set
\begin{equation}\label{eq:classic_feasible}
\mcf(X) =
\left\{
F_\mcx \in \Delta(\mathbb{R}_+^n)
:
\mathbb{E}_{x \sim F_\mcx}\!\left[\sum_{i=1}^n x_i\right] = X
\right\}.
\end{equation}
where $\Delta(S)$ denotes the set of all cumulative probability distributions over a subset $S$. The attacker analogously chooses a probability distribution $F_Y \in \mcf(Y)$.
Given mixed strategies $(F_X,F_Y)$, the expected payoff of the defender is
\begin{equation}
\pi_X(F_X,F_Y)
=
\mathbb{E}_{x\sim F_X,\;y\sim F_Y}
\big[
\mathbf{1}\{x_i \ge y_i \ \forall i\}
\big].
\end{equation}
These elements define a strategic-form game that we refer to as the \emph{classic weakest-link game}, and denote an instance of the game with $\text{WL}(X,Y,n)$.

\subsection*{Performance metrics and solution concept}

Throughout this paper, we focus on two quantities of interest regarding the performance of player $\mcx$: the \emph{max-min} and \emph{min-max} values,
\begin{equation}\label{eq:maxmin_minmax}
	\begin{aligned}
		&\max_{F_\mcx\in\mcf(X)} \min_{F_\mcy\in\mcf(Y)} \pi_\mcx(F_\mcx,F_\mcy) \\
		&\min_{F_\mcy\in\mcf(Y)} \max_{F_\mcx\in\mcf(X)}  \pi_\mcx(F_\mcx,F_\mcy).
	\end{aligned}
\end{equation}
The max-min value is the highest payoff that the defender can guarantee for itself, regardless of how the attacker plays.
The min-max value is the lowest payoff for player $\mcx$ that player $\mcy$ can guarantee, regardless of how the defender plays.
It always holds that the max-min value is no greater than the min-max value,
\begin{equation}\label{eq:SV_inequality}
	\begin{aligned}
		&\max_{F_\mcx\in\mcf(X)} \min_{F_\mcy\in\mcf(Y)} \pi_\mcx(F_\mcx,F_\mcy)  \\
		&\quad\quad \leq \min_{F_\mcy\in\mcf(Y)} \max_{F_\mcx\in\mcf(X)}  \pi_\mcx(F_\mcx,F_\mcy).
	\end{aligned}
\end{equation}
If it is the case that they are equivalent, then their common value is referred to as the \emph{equilibrium value}.
The equilibrium value of $\text{WL}(X,Y,n)$ is an established result in the literature.

\begin{theorem}[Example 5, \cite{Aghajan2024Extension}]
	Consider an instance of the classical weakest-link game $\text{WL}(X,Y,n)$. Then the equilibrium payoff to the defender is given by
\begin{equation}
L(\alpha) =
\begin{cases}
1-\dfrac{\alpha}{2}, & \alpha \le 1,\\
\dfrac{1}{2\alpha}, & \alpha > 1,
\end{cases}
\label{eq:Lalpha}
\end{equation}
where $\alpha = \frac{nY}{X}$ denotes the relative strength of the attacker.
\end{theorem}

\subsection{Networked Defense Architecture}

In this paper, we consider a variation of the classical weakest-link model where the defender deploys multiple types of defensive resources against multiple types of attacks. Unlike the classical model where a single type of defensive resource is allocated across all nodes, here multiple defensive resources contribute to protection against different attack types through a networked effectiveness structure (see Figure~\ref{fig:model}). In particular, the elements of $\mathcal{N}=\{1,\ldots,n\}$ represent distinct types of attacks or vulnerabilities faced by a single system, rather than necessarily physical locations in a network.

The defender possesses $m$ defensive resource types, labeled $\mcm = \{1,\ldots,m\}$, each with individual budgets $\mathbf{X} = (X_1,\dots,X_m)$, where $X_j > 0$ $\forall j\in\mcm$.
The defender's resource types can be allocated across attack types, but are coupled through effectiveness weights.
In particular, we define $W=[w_{j,i}] \in \mathbb{R}_+^{m\times n}$ as the effectiveness matrix for the defender's resource types, where $w_{j,i}$ represents the effectiveness of defender resource type $j$ in protecting against attack type $i$.

In this setting, we assume that the attacker has $n$ resource types, where each is individually specialized to one of the $n$ attack types (see Figure~\ref{fig:model}). This captures a setting in which the system faces distinct types of attacks, such as phishing, malware, or ransomware, where each attack type is associated with a specialized attack resource, while a defensive resource may be effective against multiple attack types.

The individual budgets for the attacker's types are given by $\mathbf{Y}=(Y_1,\dots,Y_n)$, where $Y_i>0$ $\forall i\in\mcn$.

An \emph{allocation of effort} for the defender is a matrix $\bs{x} = \{x_{i,j}\}_{i\in\mcn, j\in\mcm} \in \mbb{R}_+^{n\times m}$, where $(x_{1,j},\ldots,x_{n,j}) \in \mbb{R}_+^n$ is an allocation vector of the defender's $j$th resource type across the $n$ attack types. An allocation of effort for the attacker is an allocation vector $y = (y_1,\ldots,y_n) \in \mbb{R}_+^n$, where $y_i$ represents the attack effort associated with attack type $i$. Given allocations of effort $\bs{x},y$, the payoff to the defender is given by
\begin{equation}
	\pi_\mcx(\bs{x},y) = \mathbf{1}\left\{ \sum_{j=1}^{m} w_{j,i} x_{i,j}
    \geq y_i, \ \forall i = 1,\ldots,n \right\}.
\end{equation}
In other words, the defender possesses the weakest-link objective, where the defender successfully protects the system if and only if its effective defense is sufficient against every attack type. Otherwise, the attacker successfully compromises the system through at least one attack type. This objective is appropriate for security settings where a successful attack through any one vulnerability is sufficient to compromise the system.

\begin{figure}[t]
\centering
\includegraphics[width=0.6\linewidth]{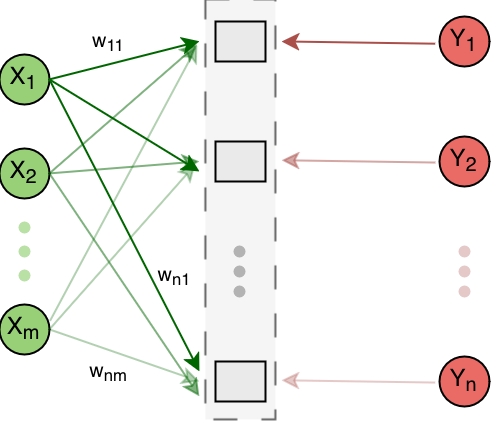}
\caption{Networked defense architecture. Multiple defensive resources contribute to protection against different attack types through heterogeneous effectiveness weights.}
\label{fig:model}
\end{figure}

Similarly to the classic setting, players are able to randomize over allocations of effort in a way that the expected allocations \emph{per resource type} satisfy the individual budget constraints. The defender chooses a distribution $F_\mcx$ over $\mathbb{R}_+^{n \times m}$ belonging to the set
\begin{equation}\label{eq:feasible_X}
\mcf_{\mcx}(\mathbf{X}) =
\left\{
\begin{aligned}
F_\mcx \in \Delta(\mathbb{R}_+^{n\times m}) \;:\;
& \mathbb{E}_{\boldsymbol{x} \sim F_\mcx}\!\left[\sum_{i=1}^n x_{i,j}\right] = X_j, \\
& \forall j=1,\dots,m
\end{aligned}
\right\}.
\end{equation}

The attacker chooses a distribution $F_\mcy$ over $\mathbb{R}_+^n$ satisfying
\begin{equation}\label{eq:feasible_Y}
\mcf_{\mcy}(\mathbf{Y}) =
\left\{
F_\mcy \in \Delta(\mathbb{R}_+^n) :
\mathbb{E}_{y \sim F_\mcy}[y_i] = Y_i,
\ \forall i=1,\dots,n
\right\}.
\end{equation}

Observe that the strategy sets $\mcf_{\mcx}(\mathbf{X}), \mcf_{\mcy}(\mathbf{Y})$ structurally differ from those given in the classic WL game~\eqref{eq:classic_feasible}, as they impose separate expected budget constraints for the individual resource types. Given mixed strategies $(F_\mcx,F_\mcy) \in \mcf_{\mcx}(\mathbf{X})\times \mcf_{\mcy}(\mathbf{Y})$, the expected payoff of the defender is
\begin{equation}
\pi_\mcx(F_\mcx,F_\mcy)
=
\mathbb{E}_{\bs{x}\sim F_\mcx,\;y\sim F_\mcy}
\bigg[
\mathbf{1}\left\{
\sum_{j=1}^{m} w_{j,i} x_{i,j} \ge y_i,\ \forall i
\right\}
\bigg].
\end{equation}

These elements define a strategic-form game that we refer to as the \emph{networked defense weakest-link game}, and denote an instance of the game with $\text{NDWL}(W,\mathbf{X},\mathbf{Y})$.

Here, the term ``networked'' refers to the heterogeneous effectiveness relationships between defensive resource types and attack types represented by $W$, rather than to a physical network topology or network dynamics.

We are again interested in characterizing the defender's security values, i.e., optimal or near-optimal values for the max-min and min-max problems in the context of the NDWL game.

\subsection{Relation to Existing Models}

The proposed NDWL game in this paper is not a \emph{direct} generalization of the classic WL game because the attacker has a unique resource type for each attack type. However, we could formulate the model in more generality, where the attacker also has its own effectiveness matrix. We leave this extension for future work. Nevertheless, the defender's effectiveness structure with multiple resource types offers a richer strategic landscape than the classic WL game, since each defensive resource type may be effective against multiple types of attacks.

The defender-side resource structure from the classic setting is recovered when there is only a single defensive resource type $m=1$ that is equally effective against all attack types, i.e., $w_{1,i}=1$ for all $i\in\mathcal{N}$.

Indeed, our formulation is a direct generalization of multi-resource weakest-link models studied in previous work \cite{aghajan2023multiresource,paarporn2025allocationheterogeneousresourcesgeneral}, where our NDWL game recovers these models when we set $m=n$ and $W$ to be a diagonal matrix, so that each defensive resource type is only effective against a single attack type. In contrast, the off-diagonal entries in $W$ allow a defensive resource type to contribute to protection against multiple attack types, thereby introducing flexibility in how defensive resources are routed across attack types. This modeling asymmetry captures settings in which attack resources are specialized by attack type, while defensive resources may provide protection against multiple types of attacks.


\section{main result}
\label{sec:result}

In this section, we detail our main results of the paper, which establish upper and lower bounds on the defender's min-max and max-min values, respectively. Before stating the Theorems, we first introduce special classes of strategies for the attacker and defender that we use in our proofs.

\begin{definition}[Best-shot attacker strategies]
\label{def:bestshot}

We define
\[
\mcf^{\mathrm{BS}}_{\mcy}(\mathbf{Y}) \subset \mcf_{\mcy}(\mathbf{Y})
\]
as the class of \emph{best-shot attacker strategies} (where $\mcf_{\mcy}(\mathbf{Y})$ is defined as in~\eqref{eq:feasible_Y}). Any strategy $F_\mcy \in \mcf^{\mathrm{BS}}_{\mcy}(\mathbf{Y})$ is completely parameterized by a pair $(\delta_Y,p)$, where $\delta_Y\in(0,1]$ and $p=(p_1,\dots,p_n)\in\Delta^n$.
An allocation of effort $y \in \mbb{R}_+^n$ is drawn from $F_\mcy$ according to the following process: with probability $1 - \delta_Y$, the attacker allocates zero effort to all attack types. With probability $\delta_Y$, it draws $U \sim \mathrm{Unif}[0,1]$, selects attack type $i$ with probability $p_i$, and allocates
\[
y_i = \frac{2Y_i}{\delta_Y p_i} U,
\qquad
y_k = 0 \quad \text{for } k \neq i.
\]
\end{definition}

Intuitively, a best-shot strategy concentrates all nonzero effort on a single, randomly selected attack type.
One can easily verify that any $F_\mcy \in \mcf^{\mathrm{BS}}_{\mcy}(\mathbf{Y})$ satisfies the mean constraints, $\mathbb{E}[y_i] = Y_i, \quad \forall i = 1,\dots,n$, as required from~\eqref{eq:feasible_Y}. The joint CDF of a best-shot strategy $F_\mcy \in \mcf^{\mathrm{BS}}_{\mcy}$ over allocations $y \in \mbb{R}_+^n$ can be expressed in the form
\begin{equation}
F_\mcy(y_1,\dots,y_n)
=
(1-\delta_Y)
+
\delta_Y \sum_{i=1}^n
p_i
\min\!\left\{
\frac{\delta_Y p_i}{2Y_i} y_i,
\,1
\right\}.
\label{eq:attacker_joint_cdf}
\end{equation}

\begin{definition}[Weakest-link defender strategies]
\label{def:weakestlink}

We define
\[
\mcf^{\mathrm{WL}}_{\mcx}(\mathbf{X}) \subset \mcf_{\mcx}(\mathbf{X})
\]
as the class of \emph{weakest-link defender strategies} (where $\mcf_{\mcx}(\mathbf{X})$ is defined as in~\eqref{eq:feasible_X}). Any strategy $F_\mcx \in \mcf^{\mathrm{WL}}_{\mcx}(\mathbf{X})$ is completely parameterized by a pair $(\delta_X,S)$, where $\delta_X \in (0,1]$ and
\[
S \in \mathcal{S} := \left\{ S \in \mathbb{R}_+^{m \times n} \;:\; \sum_{i=1}^n s_{j,i} = 1,\ \forall j = 1,\dots,m \right\}.
\]
An allocation of effort $\bs{x} \in \mbb{R}_+^{n\times m}$ is drawn from $F_\mcx$ according to the following process:
with probability $1 - \delta_X$, the defender allocates zero resources, i.e., $x_{i,j}=0$ for all $i,j$. With probability $\delta_X$, it draws $U \sim \mathrm{Unif}[0,1]$ and allocates
\[
x_{i,j} = \frac{2 X_j s_{j,i}}{\delta_X} U,
\qquad \forall i=1,\dots,n,\; j=1,\dots,m.
\]
We refer to a matrix $S \in \mcal{S}$ as a \emph{routing matrix}.
\end{definition}

Intuitively, all resource allocations $x_{i,j}$ are correlated through a single random variable $U$, and the routing matrix $S$ determines the proportions of each defender resource type that are allocated to each attack type.
One can verify that $F_\mcx \in \mcf_{\mcx}(\mathbf{X})$, since the induced allocation satisfies the mean constraints, $\mathbb{E}[\sum_{i=1}^n x_{i,j}] = X_j, \quad \forall j\in\mcal{M}$, as required by~\eqref{eq:feasible_X}. The joint CDF of a weakest-link strategy $F_\mcx \in \mcf^{\mathrm{WL}}_{\mcx}$ over allocations $\bs{x} \in \mbb{R}_+^{n\times m}$ can be expressed as
\begin{equation}
F_\mcx(\bs{x})
=
(1-\delta_X)
+
\delta_X
\min_{i=1,\dots,n \atop j=1,\dots,m}
\left\{
\frac{\delta_X}{2 X_j s_{j,i}} x_{i,j},
\,1
\right\}.
\label{eq:defender_joint_cdf}
\end{equation}

\subsection{Upper Bound on the Defender's Success Probability}
\label{sec:upper}

We now establish an upper bound on the defender's min-max value in the result below.
For a routing matrix $S\in\mathcal{S}$, define the effective expected defense level against attack type $i$ as
\begin{equation}
z_i(S)
:=
\sum_{j=1}^m w_{j,i}X_j s_{j,i},
\qquad i=1,\dots,n.
\label{eq:zi}
\ref{eq:zi}
\end{equation}

\begin{theorem}[Upper bound]
\label{thm:upper}
Consider the NDWL$(W,\mathbf{X},\mathbf{Y})$ game defined in Section~\ref{sec:model}.
Define
\begin{equation}
g(p,S)
:=
\sum_{i=1}^n p_i^2\frac{z_i(S)}{Y_i},
\label{eq:g_general}
\end{equation}
and
\begin{equation}
g^\star
:=
\min_{p\in\Delta^n}
\max_{S\in\mathcal{S}}
g(p,S).
\end{equation}

Let
\begin{equation}
\mathrm{UB} = L\!\left(\frac{1}{g^\star}\right)
=
\begin{cases}
\dfrac{g^\star}{2}, & g^\star \le 1, \\[6pt]
1 - \dfrac{1}{2 g^\star}, & g^\star \ge 1.
\end{cases}
\end{equation}

Then the defender's min-max value satisfies
\begin{equation}
\min_{F_\mcy \in \mcf_{\mcy}(\mathbf{Y})}
\max_{F_\mcx \in \mcf_{\mcx}(\mathbf{X})}
\pi_X(F_\mcx,F_\mcy)
\;\le\;
\mathrm{UB}.
\end{equation}
\end{theorem}

\begin{proof}
Suppose $F_\mcx \in \mcf_{\mcx}(\mathbf{X})$. We begin by evaluating the defender's payoff against an attacker strategy $F_\mcy \in \mcf_{\mcy}^{\mathrm{BS}}(\mathbf{Y})$, which has parameters $(\delta_Y,p)$ (from Definition~\ref{def:bestshot}).
Using its CDF form~\eqref{eq:attacker_joint_cdf}, we can write
\begin{equation}
\begin{aligned}
\pi_X(F_\mcx,F_\mcy)
&=
\mathbb{E}_{x\sim F_\mcx}\bigg[
F_\mcy\!\left(
\sum_{j=1}^m w_{j,1}x_{1,j},\dots,\sum_{j=1}^m w_{j,n}x_{n,j}
\right)
\bigg] \\
&=
\mathbb{E}_{x\sim F_\mcx}
\Big[
(1-\delta_Y)
+
\delta_Y \sum_{i=1}^n p_i\, \min\{M_i(x),1\}
\Big].
\end{aligned}
\end{equation}

Define
\begin{equation}
M_i(x):=\frac{\delta_Y p_i}{2Y_i}\sum_{j=1}^m w_{j,i}x_{i,j}.
\end{equation}
Then, for every realization $x$, we have
\[
\min\{M_i(x),1\}\le M_i(x).
\]
Taking expectation with respect to $x\sim F_\mcx$ yields
\begin{equation}
\mathbb{E}_{x\sim F_\mcx}[\min\{M_i(x),1\}]
\le
\mathbb{E}_{x\sim F_\mcx}[M_i(x)].
\end{equation}
Substituting the definition of $M_i(x)$ and using linearity of expectation, we obtain
\begin{equation}
\mathbb{E}_{x\sim F_\mcx}[M_i(x)]
=
\frac{\delta_Y p_i}{2Y_i}\,
\mathbb{E}_{x\sim F_\mcx}\bigg[\sum_{j=1}^m w_{j,i}x_{i,j}\bigg].
\end{equation}
Hence,
\begin{equation}
\mathbb{E}_{x\sim F_\mcx}[\min\{M_i(x),1\}]
\le
\frac{\delta_Y p_i}{2Y_i}\,
\mathbb{E}_{x\sim F_\mcx}\bigg[\sum_{j=1}^m w_{j,i}x_{i,j}\bigg].
\end{equation}

Thus,
\begin{equation}
\pi_X(F_\mcx,F_\mcy)
\le
1-\delta_Y
+
\frac{\delta_Y^2}{2}
\sum_{i=1}^n p_i^2
\frac{
\mathbb{E}_{x\sim F_\mcx}
\big[\sum_{j=1}^m w_{j,i}x_{i,j}\big]
}{Y_i}.
\label{eq:upper_before_S}
\end{equation}

We now express the expected defense levels through an induced routing structure:
\begin{equation}
\mathbb{E}_{x\sim F_\mcx}\bigg[\sum_{j=1}^m w_{j,i}x_{i,j}\bigg]
=
\sum_{j=1}^m w_{j,i}\,
\mathbb{E}_{x\sim F_\mcx}[x_{i,j}].
\end{equation}

For each resource $j$, define
\begin{equation}
s_{j,i}:=
\frac{\mathbb{E}_{x\sim F_\mcx}[x_{i,j}]}{X_j},
\qquad i=1,\dots,n,\ j=1,\dots,m.
\end{equation}
Hence $S=[s_{j,i}]$ is a feasible routing matrix in the sense of Definition~\ref{def:weakestlink}.
Then,
\begin{equation}
\mathbb{E}_{x\sim F_\mcx}\bigg[\sum_{j=1}^m w_{j,i}x_{i,j}\bigg]
=
z_i(S).
\end{equation}

Since every mixed strategy $F_\mcx$ induces a feasible routing matrix $S$, the maximization over $F_\mcx$ can be upper bounded by a maximization over feasible routing matrices.

Substituting into~\eqref{eq:upper_before_S} yields
\begin{equation}
\pi_X(F_\mcx,F_\mcy)
\le
1-\delta_Y
+
\frac{\delta_Y^2}{2}
\sum_{i=1}^n p_i^2\frac{z_i(S)}{Y_i}.
\end{equation}

Thus,
\begin{equation}
\pi_X(F_\mcx,F_\mcy)
\le
1-\delta_Y+\frac{\delta_Y^2}{2}g(p,S).
\end{equation}

Minimizing over $\delta_Y\in[0,1]$ yields
\begin{equation}
L\!\left(\frac{1}{g}\right)
=
\begin{cases}
\dfrac{g}{2}, & g\le 1, \\[6pt]
1-\dfrac{1}{2g}, & g\ge 1.
\end{cases}
\end{equation}
Indeed, the minimum is reached at $\delta_Y=1$ when $g\le 1$ and at $\delta_Y=1/g$ when $g\ge 1$. Therefore,
\begin{equation}
\pi_X(F_\mcx,F_\mcy)
\le
L\!\left(\frac{1}{g(p,S)}\right).
\end{equation}

Since $L(1/g)$ increases in $g>0$, minimizing over $p\in\Delta^n$ and maximizing over feasible $S$ yields
\[
\mathrm{UB}
:=
L\!\left(
\frac{1}{\min_{p\in\Delta^n}\max_{S\in\mathcal{S}} g(p,S)}
\right).
\]
Hence,
\begin{equation}
\min_{F_\mcy \in \mcf_{\mcy}^{\mathrm{BS}}(\mathbf{Y})}
\max_{F_\mcx \in \mcf_{\mcx}(\mathbf{X})}
\pi_X(F_\mcx,F_\mcy)
\;\le\;
\mathrm{UB}.
\end{equation}

Finally, since this derivation is based on the attacker using a best-shot strategy $F_\mcy\in\mcf_{\mcy}^{\mathrm{BS}}(\mathbf{Y})$, where $\mcf_{\mcy}^{\mathrm{BS}}(\mathbf{Y}) \subset \mcf_{\mcy}(\mathbf{Y})$, we have
\begin{align}
\min_{F_\mcy\in\mcf_{\mcy}(\mathbf{Y})}
\max_{F_\mcx\in\mcf_{\mcx}(\mathbf{X})}
\pi_X(F_\mcx,F_\mcy)
&\le
\min_{F_\mcy\in\mcf_{\mcy}^{\mathrm{BS}}(\mathbf{Y})}
\max_{F_\mcx\in\mcf_{\mcx}(\mathbf{X})}
\pi_X(F_\mcx,F_\mcy) \nonumber\\
&\le
\mathrm{UB}.
\end{align}
\end{proof}
The optimization problem that characterizes $g^\star$ exhibits a favorable structural form: for fixed $S$, $g(\cdot,S)$ is a weighted sum-of-squares (and thus convex) in $p$, since $z_i(S)/Y_i\geq 0$, and for fixed $p$, it is affine in $S$. Moreover, because $\Delta^n$ and the set of feasible routing matrices are convex, the function $\max_{S\in\mathcal{S}}g(p,S)$ is convex in $p$ as a pointwise maximum of convex functions. Hence, the upper-bound problem is a convex minimization over the simplex.

\subsection{Lower Bound on the Defender's Success Probability}
\label{sec:lower}

We next establish a lower bound on the defender's max-min value in the result below.

\begin{theorem}[Lower bound]
\label{thm:lower}
Consider the NDWL$(W,\mathbf{X},\mathbf{Y})$ game defined in Section~\ref{sec:model}. 
Define
\begin{equation}
\alpha(S)
:=
\sum_{i=1}^n \frac{Y_i}{z_i(S)},
\label{eq:alpha}
\end{equation}
and
\begin{equation}
\alpha^\star
:=
\min_{S\in\mathcal{S}}\alpha(S).
\label{eq:alpha_star}
\end{equation}

and let
\begin{equation}
LB = L(\alpha^\star)
=
\begin{cases}
1-\dfrac{\alpha^\star}{2}, & \alpha^\star\le 1,\\[6pt]
\dfrac{1}{2\alpha^\star}, & \alpha^\star\ge 1.
\end{cases}
\label{eq:lower_bound_final}
\end{equation}

Then the defender's max-min value satisfies
\begin{equation}
\max_{F_\mcx \in \mcf_{\mcx}(\mathbf{X})}
\min_{F_\mcy \in \mcf_{\mcy}(\mathbf{Y})}
\pi_X(F_\mcx,F_\mcy)
\;\ge\;
\mathrm{LB}.
\end{equation}
\end{theorem}

\begin{proof}
Suppose $F_\mcy \in \mcf_{\mcy}(\mathbf{Y})$, and suppose the defender employs a weakest-link strategy $F_\mcx \in \mcf_{\mcx}^{\mathrm{WL}}(\mathbf{X})$, which has associated parameters $(\delta_X,S)$ (from Definition~\ref{def:weakestlink}).

Under the weakest-link strategy $F_\mcx$, a realization of the effective defense level against attack type $i$ is
\[
\tilde z_i=\frac{2}{\delta_X}z_i(S)\,U,
\qquad U\sim \mathrm{Unif}[0,1],
\]
where
\[
z_i(S)=\sum_{j=1}^m w_{j,i}X_j s_{j,i}.
\]
Fix an attack realization $y=(y_1,\dots,y_n)$ and define
\begin{equation}
A_i(y):=\{\tilde z_i < y_i\},
\qquad i\in \mathcal{N}.
\end{equation}
The attacker succeeds whenever the defense is insufficient against at least one attack type, i.e.,
\begin{equation}
\bigcup_{i\in \mathcal{N}} A_i(y).
\end{equation}

Hence, for fixed $y$, the attacker's success probability is
\begin{equation}
\Pr\!\left(\bigcup_{i\in \mathcal{N}} A_i(y)\right).
\end{equation}

Applying the inclusion--exclusion principle over subsets $T \subseteq \mathcal{N}$ with $|T|\ge 1$ gives
\begin{equation}
\Pr\!\left(\bigcup_{i\in \mathcal{N}} A_i(y)\right)
=
\sum_{T \subseteq \mathcal{N},\ |T|\ge 1}
(-1)^{|T|+1}
\Pr\!\left(\bigcap_{i\in T} A_i(y)\right).
\end{equation}

Define
\begin{equation}
F_{\mcx,T}(y)
:=
\Pr\!\left(\bigcap_{i\in T} \{\tilde z_i < y_i\}\right).
\end{equation}

Taking expectation with respect to $y\sim F_\mcy$, we can express the attacker's expected payoff as
\begin{equation}
\pi_\mcy(F_\mcx,F_\mcy)
=
\mathbb{E}_{y\sim F_\mcy}
\left[
\sum_{T \subseteq \mathcal{N},\ |T|\ge 1}
(-1)^{|T|+1} F_{\mcx,T}(y)
\right].
\label{eq:lower_ie}
\end{equation}

Under the weakest-link strategy, the quantity $F_{\mcx,T}(y)$ can be obtained from the joint CDF in~\eqref{eq:defender_joint_cdf}. In particular, for each nonempty subset $T$,
\begin{equation}
F_{\mcx,T}(y)
=
(1-\delta_X)
+
\delta_X
\Pr\!\left(
U <
\min_{i\in T}
\frac{\delta_X}{2z_i(S)} y_i
\right).
\end{equation}

Define
\begin{equation}
m_i(y_i)
:=
\min
\left\{
\frac{\delta_X}{2z_i(S)} y_i,\,
1
\right\},
\qquad i\in \mathcal{N}.
\label{eq:mi_def}
\end{equation}

Since $U\sim\mathrm{Unif}[0,1]$, this becomes
\begin{equation}
F_{\mcx,T}(y)
=
(1-\delta_X)
+
\delta_X
\min_{i\in T} m_i(y_i).
\end{equation}

Substituting into~\eqref{eq:lower_ie} and simplifying using the inclusion--exclusion identity yields
\begin{align}
\pi_\mcy(F_\mcx,F_\mcy)
&=
\mathbb{E}_{y\sim F_\mcy}
\!\left[
(1-\delta_X)
\sum_{\substack{T \subseteq \mathcal{N}\\ |T|\ge 1}}
(-1)^{|T|+1}
\right]
\nonumber\\
&\quad+
\mathbb{E}_{y\sim F_\mcy}
\!\left[
\delta_X
\sum_{\substack{T \subseteq \mathcal{N}\\ |T|\ge 1}}
(-1)^{|T|+1}
\min_{i\in T} m_i(y_i)
\right].
\end{align}

To proceed, we leverage the following technical identities:
\begin{equation}
\sum_{T \subseteq \mathcal{N},\ |T|\ge 1}
(-1)^{|T|+1}
=
1,
\end{equation}
and
\begin{equation}
\sum_{T \subseteq \mathcal{N},\ |T|\ge 1}
(-1)^{|T|+1}
\min_{i\in T} m_i
=
\max_{i\in \mathcal{N}} m_i,
\label{eq:max_identity}
\end{equation}
which holds for any collection of nonnegative scalars
$\{m_i\}_{i\in \mathcal{N}}$
(see \cite{paarporn2025allocationheterogeneousresourcesgeneral}, Appendix A).
We then obtain
\begin{equation}
\pi_\mcy(F_\mcx,F_\mcy)
=
\mathbb{E}_{y\sim F_\mcy}
\big[
(1-\delta_X)
+
\delta_X
\max_{i\in \mathcal{N}} m_i(y_i)
\big].
\label{eq:piY_max}
\end{equation}

Since
$\max_i m_i(y_i) \le \sum_{i\in \mathcal{N}} m_i(y_i)$,
it follows that
\begin{equation}
\pi_\mcy(F_\mcx,F_\mcy)
\le
(1-\delta_X)
+
\delta_X
\sum_{i\in \mathcal{N}}
\mathbb{E}_{y\sim F_\mcy}[m_i(y_i)].
\end{equation}

From~\eqref{eq:mi_def},
\begin{equation}
m_i(y_i)
\le
\frac{\delta_X}{2z_i(S)} y_i.
\end{equation}
Hence,
\begin{equation}
\mathbb{E}_{y\sim F_\mcy}[m_i(y_i)]
\le
\frac{\delta_X}{2z_i(S)}
\,\mathbb{E}_{y\sim F_\mcy}[y_i]
=
\frac{\delta_X}{2z_i(S)}Y_i.
\end{equation}

Therefore,
\begin{equation}
\pi_\mcy(F_\mcx,F_\mcy)
\le
(1-\delta_X)
+
\frac{\delta_X^2}{2}
\sum_{i\in \mathcal{N}} \frac{Y_i}{z_i(S)}.
\end{equation}

This shows that the lower bound depends on the routing matrix $S$ only through the aggregate quantities $z_i(S)$, reducing the problem to minimizing $\alpha(S)$ over feasible routings.

Since $\pi_\mcx=1-\pi_\mcy$, we obtain
\begin{equation}
\pi_\mcx(F_\mcx,F_\mcy)
\ge
\delta_X
\left(
1-\frac{\delta_X}{2}\alpha(S)
\right).
\end{equation}

Maximizing over $\delta_X\in[0,1]$ yields
\begin{equation}
\pi_X(F_\mcx,F_\mcy)
\ge
L(\alpha(S)).
\end{equation}

Finally, minimizing $\alpha(S)$ over feasible routing matrices $S$ yields
\begin{equation}
\max_{F_\mcx \in \mcf_{\mcx}^{\mathrm{WL}}(\mathbf{X})}
\min_{F_\mcy \in \mcf_{\mcy}(\mathbf{Y})}
\pi_X
\;\ge\;
\mathrm{LB},
\end{equation}
where
\[
\mathrm{LB}
:=
L(\alpha^\star),
\qquad
\alpha^\star
=
\min_{S\in\mathcal{S}}\alpha(S).
\]

Because this derivation is based on the defender using a weakest-link strategy $F_\mcx \in \mcf_{\mcx}^{\mathrm{WL}}(\mathbf{X})$, where
\(
\mcf_{\mcx}^{\mathrm{WL}}(\mathbf{X})
\subset
\mcf_{\mcx}(\mathbf{X}),
\)
we have
\begin{align}
\max_{F_\mcx \in\mcf_{\mcx}(\mathbf{X})}
\min_{F_\mcy \in\mcf_{\mcy}(\mathbf{Y})}
\pi_X
&\ge
\max_{F_\mcx \in\mcf_{\mcx}^{\mathrm{WL}}(\mathbf{X})}
\min_{F_\mcy \in\mcf_{\mcy}(\mathbf{Y})}
\pi_X
\nonumber\\
&\ge
\mathrm{LB}.
\end{align}
\end{proof}

The optimization problem defining the lower bound has a particularly favorable structure. The objective function
\[
\alpha(S)
=
\sum_{i=1}^n \frac{Y_i}{z_i(S)}
\]
is convex in $S$, since each term $Y_i/z_i(S)$ is convex over $z_i(S)>0$, and $z_i(S)$ is affine in $S$.

In addition, the set of admissible routing matrices is convex, since it is defined by linear constraints. Consequently, minimizing $\alpha(S)$ over this set is a convex program. This allows us to compute the lower bound to global optimality using standard convex optimization techniques.


\subsection{Numerical Validation for Multiple Attack Types}

\noindent
Recall that the lower and upper bounds satisfy
\begin{equation}
\mathrm{LB}
\;\le\;
\max_{\mcf_{\mcx}(\mathbf{X})}\min_{\mcf_{\mcy}(\mathbf{Y})}\pi_X
\;\le\;
\min_{\mcf_{\mcy}(\mathbf{Y})}\max_{\mcf_{\mcx}(\mathbf{X})}\pi_X
\;\le\;
\mathrm{UB}.
\end{equation}

An analytical proof of equality between the upper and lower bounds has not yet been established for general \(n\). We thus present numerical evidence of the tightness of our bounds.

We generate a random instance with $n=5$ attack types and vary the budgets of the defender by a multiplicative parameter $\lambda$, i.e. $\mathbf{X} = \lambda \mathbf{X}_{\mathrm{base}}$ while holding the attacker budgets $\mathbf{Y}$ constant. For each such budget choice, we compute optimized versions of the upper and lower bounds (denoted $\mathrm{UB}$ and $\mathrm{LB}$ respectively) and compute the gap between them
\[
\mathrm{gap} := \mathrm{UB} - \mathrm{LB}.
\]

The upper- and lower-bound optimization problems are solved numerically in CVX on MATLAB. As detailed in Section~\ref{sec:result}, both optimization problems are convex. In particular, the upper-bound problem is a convex minimization problem over the simplex, since the objective $g(p,S)$ is convex in $p$ and affine in $S$, so maximizing over $S$ produces a convex objective in terms of $p$ only. The lower-bound problem is also convex, since the objective $\alpha(S)$ is convex in $S$ and the set of feasible routing matrices is convex. 

\begin{figure}[!t]
\centering
\includegraphics[width=0.48\textwidth]{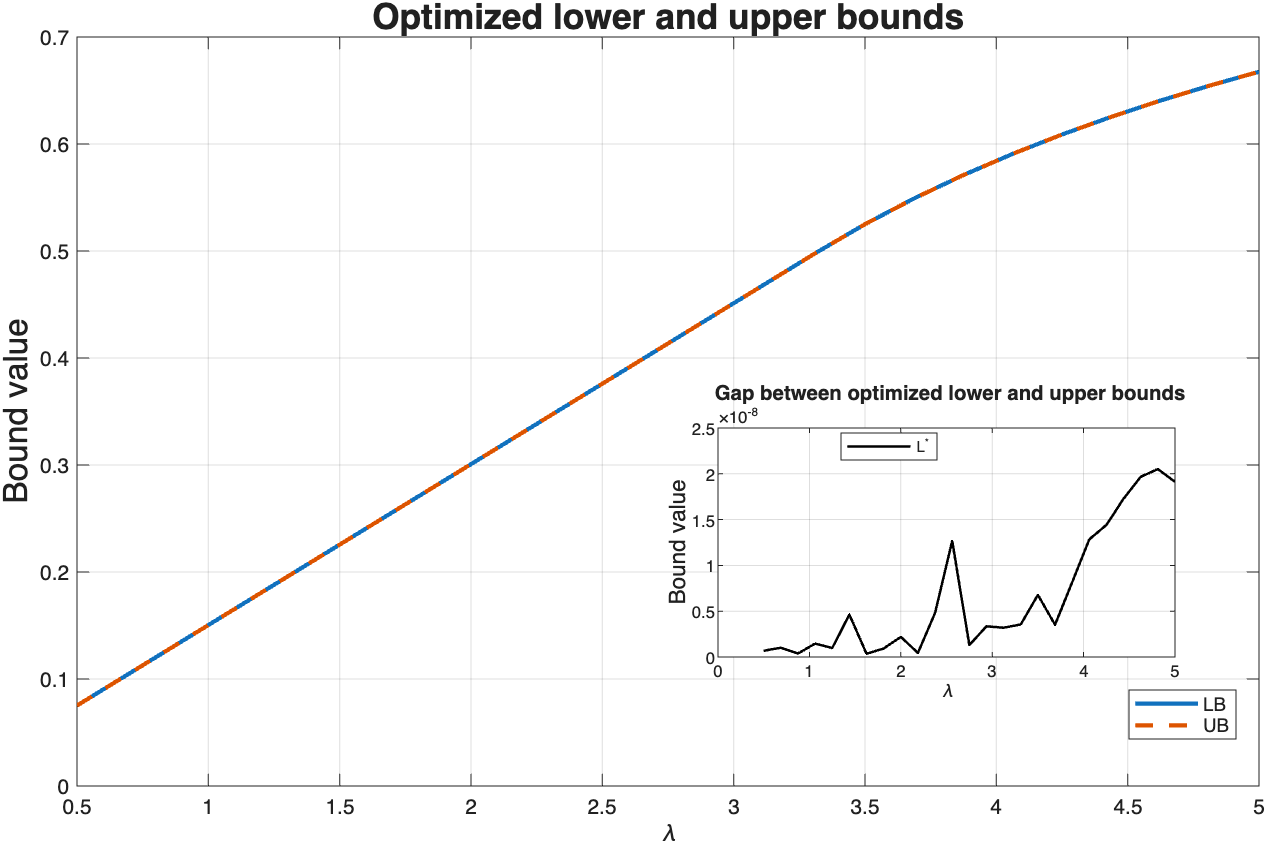}
\caption{Optimized upper and lower bounds $UB$ and $LB$ as a function of the scaling parameter $\lambda$ for $n=5$ attack types. Here, $\lambda$ scales the defender's resource vector according to $\mathbf{X} = \lambda \mathbf{X}_{\mathrm{base}}$. The two curves are visually indistinguishable. The inset shows the corresponding gap $UB - LB$, which remains below $3 \times 10^{-8}$ across all tested parameter values.}
\label{fig:multinode_bounds}
\end{figure}

As shown in Fig.~\ref{fig:multinode_bounds}, the upper and lower bounds are indistinguishable across the entire parameter range. The inset reveals that the gap remains on the order of numerical precision, never exceeding $3 \times 10^{-8}$.

These observations provide strong numerical evidence that the equality $UB = LB$, which is established analytically for the two-atack-type case in Section~\ref{sec:two-nodes}, may extend to the general $n$-attack-type setting.

\begin{figure*}[t]
\centering

\begin{subfigure}{0.24\linewidth}
\includegraphics[width=\linewidth]{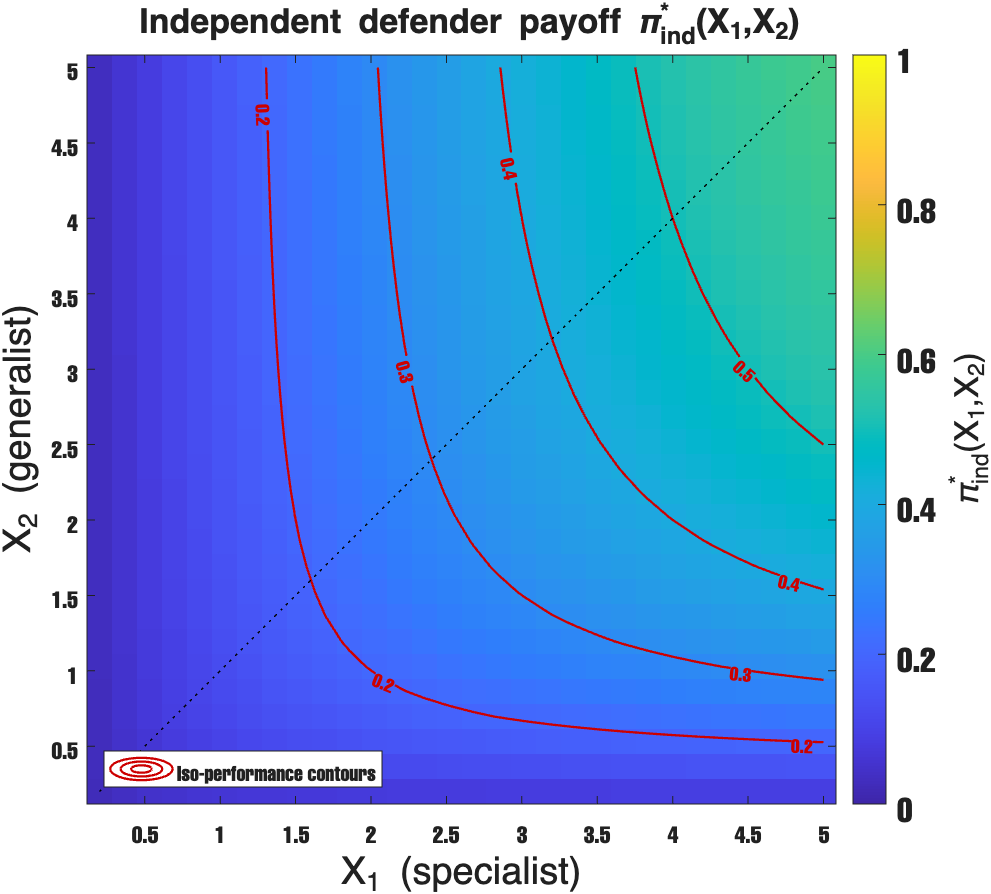}
\caption{Independent}
\end{subfigure}
\hfill
\begin{subfigure}{0.24\linewidth}
\includegraphics[width=\linewidth]{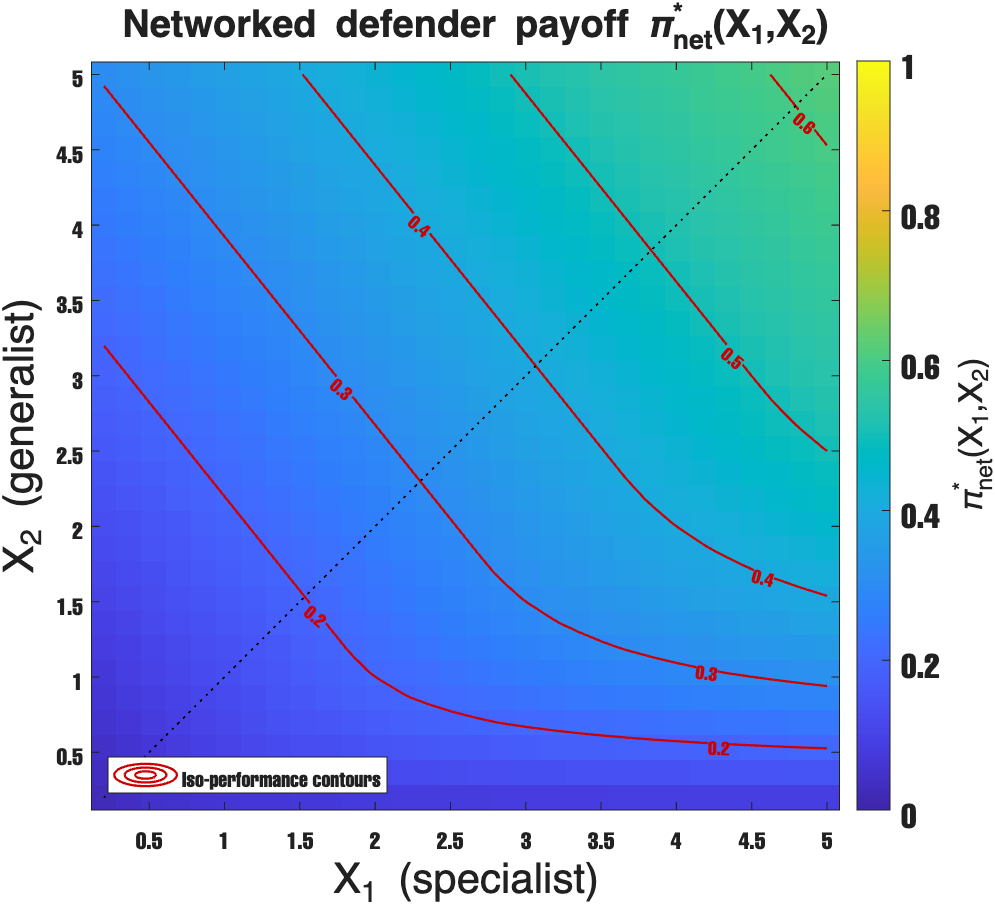}
\caption{Networked}
\end{subfigure}
\hfill
\begin{subfigure}{0.24\linewidth}
\includegraphics[width=\linewidth]{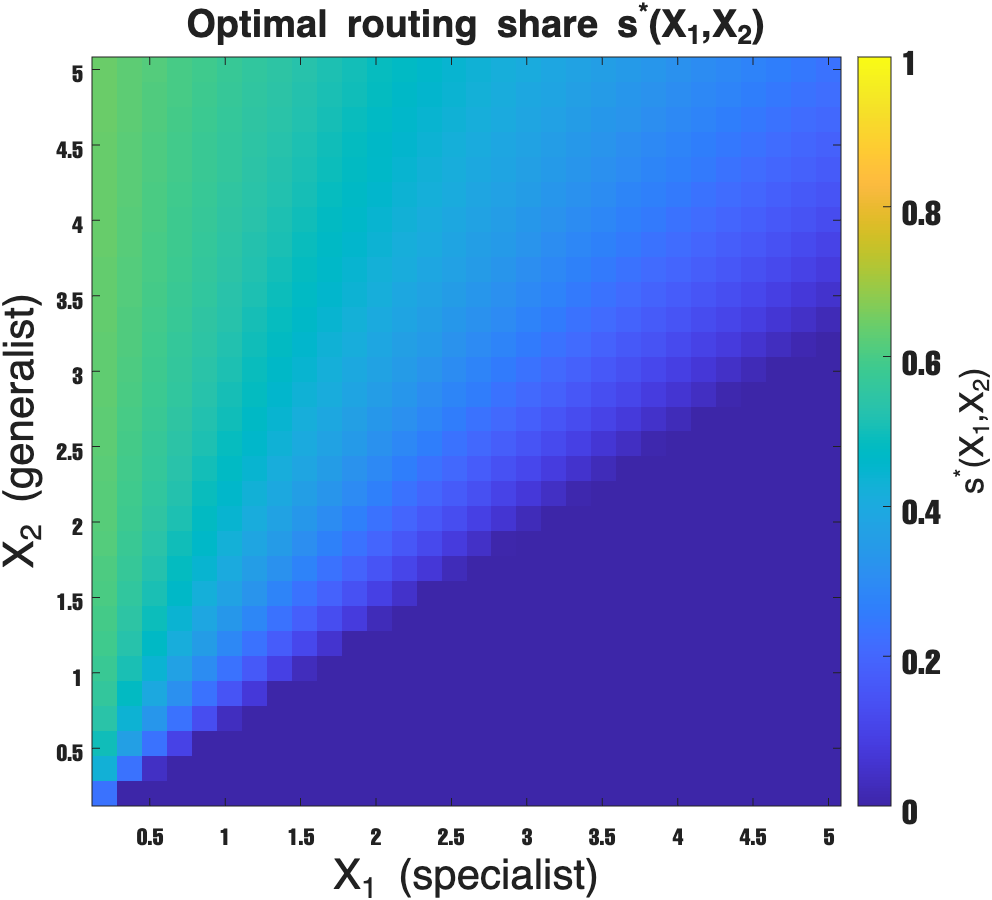}
\caption{Routing share (Networked)}
\end{subfigure}
\hfill
\begin{subfigure}{0.24\linewidth}
\includegraphics[width=\linewidth]{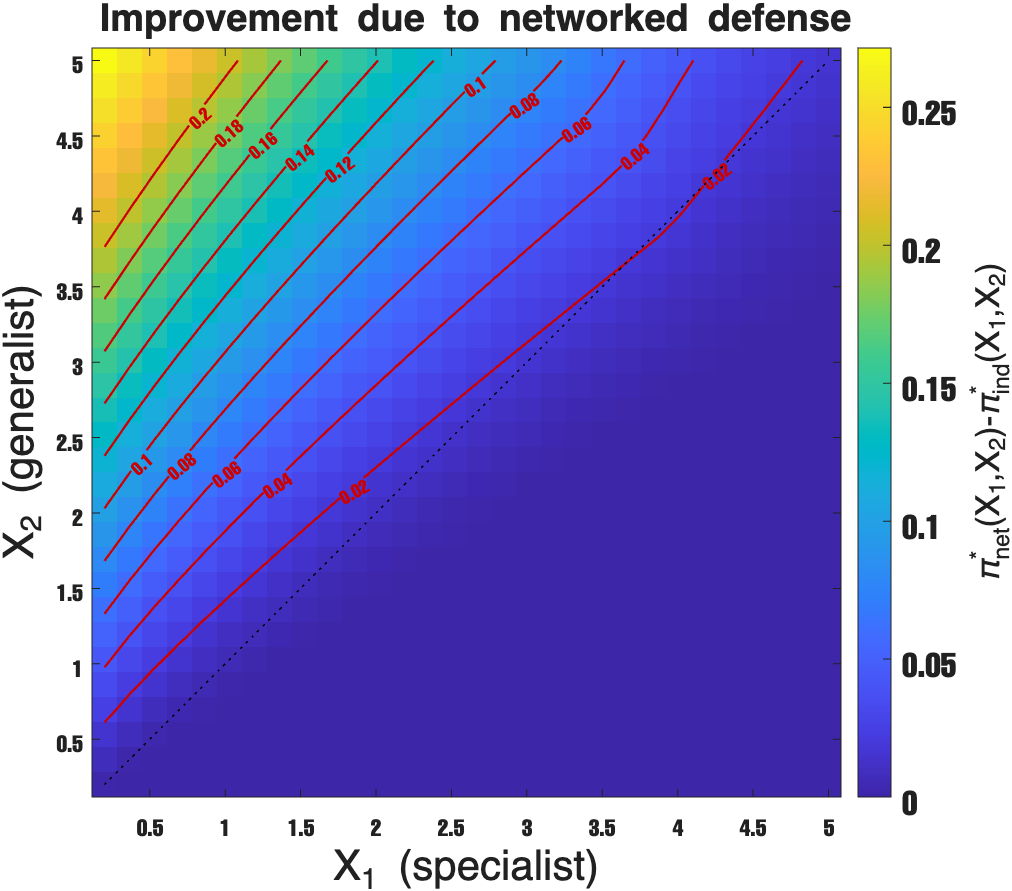}
\caption{Improvement}
\end{subfigure}

\caption{
Comparison between independent and networked defense architectures.
Panels (a) and (b) show the defender equilibrium payoff
$\pi_X^\ast(X_1,X_2)$ under the independent and networked cases,
respectively. Panel (c) shows the optimal routing share
$s^\ast(X_1,X_2)$ in the networked architecture.
Panel (d) displays the resulting performance improvement
$\pi_{\text{net}}^\ast(X_1,X_2)-\pi_{\text{ind}}^\ast(X_1,X_2)$.
}
\label{fig:networked_vs_independent}
\end{figure*}


\section{Equilibrium Value for Two Attack Types}
\label{sec:two-nodes}

In this section, we perform an in-depth analysis of a two-attack-type instance of the NDWL model with two types of defensive resources, i.e., $|\mcal{N}| = |\mcal{M}| = 2$.
The defender has an attack-specific (specialist) resource of size $X_1$ that can only protect against attack type~1,
and a flexible (generalist) resource of size $X_2$ that can be allocated across both attack types.
In this setting, we are able to characterize closed-form expressions for the upper and lower bounds derived in Section~\ref{sec:result}. Moreover, we find that these coincide, yielding an exact equilibrium value of the game.

\begin{theorem}[Exact value for two attack types in the NDWL model]
\label{thm:two_node_exact}
Consider NDWL$(W,\mathbf{X},\mathbf{Y})$, where
\[
W=\begin{bmatrix}1&0\\ w_{21}&w_{22}\end{bmatrix},
\qquad
w_{21},w_{22}>0.
\]
Then the equilibrium value is given by
\begin{equation}
V
=
L(\alpha^\star)
=
\begin{cases}
1-\dfrac{\alpha^\star}{2}, & \alpha^\star\le 1,\\[6pt]
\dfrac{1}{2\alpha^\star}, & \alpha^\star\ge 1,
\end{cases}
\end{equation}
where
\[
\alpha^\star
=
\frac{Y_1}{X_1+s^\star w_{21}X_2}
+
\frac{Y_2}{(1-s^\star)w_{22}X_2}.
\]
\end{theorem}

\begin{proof}
We first calculate the lower bound characterized in Theorem~\ref{thm:lower}. Here, the routing matrix simplifies to a scalar $s \in [0,1]$ that denotes the fraction of the generalist resource (type 2) allocated to attack type~1,
so that $(1-s)$ is allocated to attack type~2.
The induced expected defense levels against each attack type are given by
\begin{equation}
z_1 = X_1 + s\,w_{21}X_2,
\qquad
z_2 = (1-s)\,w_{22}X_2.
\label{eq:two_node_z}
\end{equation}
The aggregate quantity $\alpha(s)$~\eqref{eq:alpha} reduces to
\begin{equation}
\alpha(s)
=
\frac{Y_1}{X_1+s\,w_{21}X_2}
+
\frac{Y_2}{(1-s)\,w_{22}X_2}.
\label{eq:alpha_s}
\end{equation}

The defender chooses $s \in [0,1]$ to minimize $\alpha(s)$. Differentiating~\eqref{eq:alpha_s} gives
\begin{equation}
\alpha'(s)
=
-\frac{Y_1 w_{21}X_2}{(X_1+s\,w_{21}X_2)^2}
+
\frac{Y_2 w_{22}X_2}{((1-s)\,w_{22}X_2)^2}.
\end{equation}
Given the constraint $s\in[0,1]$, the optimal routing parameter is given by
\begin{equation}
\begin{aligned}
s^\star
=
\min\!\Bigg\{
1,\;
\max\!\Bigg\{
0,\;
\frac{
\sqrt{Y_1 w_{21}}\,X_2-\sqrt{Y_2/w_{22}}\,X_1
}{
X_2\Big(\sqrt{Y_1 w_{21}}+w_{21}\sqrt{Y_2/w_{22}}\Big)
}
\Bigg\}
\Bigg\}.
\end{aligned}
\end{equation}
The lower bound from Theorem~\ref{thm:lower} can then be expressed as $L(\alpha^\star)$, where $\alpha^\star := \alpha(s^\star)$.

We next calculate the upper bound characterized in Theorem~\ref{thm:upper}.
Recall the definition of $g(p,S)$ in~\eqref{eq:g_general}, which in this case can be expressed as
\begin{equation}
g(p,s)
=
p_1^2 \frac{z_1}{Y_1}
+
p_2^2 \frac{z_2}{Y_2}.
\label{eq:gps_two}
\end{equation}
Minimizing~\eqref{eq:gps_two} over $p$ for fixed $s$ yields
\begin{equation}
p_1^\star
=
\frac{Y_1/z_1}{Y_1/z_1+Y_2/z_2},
\qquad
p_2^\star
=
\frac{Y_2/z_2}{Y_1/z_1+Y_2/z_2},
\end{equation}
where $z_1$ and $z_2$ are given by~\eqref{eq:two_node_z}. Substituting these probabilities into~\eqref{eq:gps_two} gives
\begin{equation}
g(p^\star,s)
=
\frac{1}{\frac{Y_1}{z_1}+\frac{Y_2}{z_2}}
=
\frac{1}{\alpha(s)}.
\label{eq:g_alpha_relation}
\end{equation}

Since $g(p,s)$ is convex in $p$ and affine in $s$, the minimax equality applies, so
\begin{equation}
\min_{p\in\Delta^2}\max_{s\in[0,1]}g(p,s)
=
\max_{s\in[0,1]}\min_{p\in\Delta^2}g(p,s).
\end{equation}
Therefore, maximizing $g(p^\star,s)$ over $s$ is equivalent to minimizing $\alpha(s)$, and thus we obtain
\begin{equation}
g^\star
=
\max_{s\in[0,1]} g(p^\star,s)
=
\frac{1}{\alpha^\star}.
\end{equation}
Substituting into the optimized upper bound yields
\begin{equation}
\mathrm{UB}
=
L\!\left(\frac{1}{g^\star}\right)
=
L(\alpha^\star)
=
\mathrm{LB}.
\end{equation}
We see that the upper and lower bounds coincide, and thus we obtain the equilibrium value of the game.
\end{proof}

\subsection*{Numerical illustration for two attack types}

We illustrate this result by comparing the proposed NDWL architecture with an independent-defense benchmark from the existing literature \cite{paarporn2025allocationheterogeneousresourcesgeneral}, in which the weight matrix $W$ is diagonal. The attacker budgets are fixed at
$Y_1=3$ and $Y_2=1$,
so attack type~1 receives the larger attack budget. The defender resources $(X_1,X_2)$ are varied, where $X_1$ is a specialist resource and $X_2$ is a generalist resource.

In the independent benchmark (see, e.g.,~\cite{paarporn2025allocationheterogeneousresourcesgeneral}), the generalist resource protects only against attack type~2, corresponding to
$w_{21}=0$ and $w_{22}=1$,
whereas in the NDWL model the generalist resource can also contribute to protection against attack type~1, with
$w_{21}=0.8$ and $w_{22}=1$.

Figure~\ref{fig:networked_vs_independent} illustrates the impact of the networked architecture. The equilibrium payoff under the independent benchmark is plotted in panel (a). The analogous payoff under NDWL is plotted in panel (b). Panel (c) plots the resulting optimal routing share $s^\ast(X_1,X_2)$. Panel (d) plots the resulting performance improvement
\[
\pi_{\mathrm{NDWL}}^\ast(X_1,X_2)
-
\pi_{\mathrm{ind}}^\ast(X_1,X_2).
\]

These results demonstrate the primary advantage of the NDWL architecture: flexibility in routing resources allows the defender to allocate more protection effort toward the more heavily attacked type, leading to strictly improved performance. The largest improvement occurs when specialist resources are limited and the defender must utilize the generalist resource, highlighting the value of flexible routing to reallocate protection toward the attack type receiving the larger attack budget.

\section{CONCLUSIONS}

The objective of this paper was to analyze the NDWL model. We derived a complete characterization of the equilibrium value when there are two attack types. In particular, we proved that the optimized upper and lower bounds coincide and have a closed-form expression which reveals an intuitive relationship between the optimal routing decision and the aggregate vulnerability of the system. Our analysis elucidates the structural importance of networked resource allocation: flexibility in routing generalist resources across attack types allows the defender to adapt to asymmetric attack budgets and can strictly outperform the independent benchmark. For general $n$, we leave open the problem of proving analytically that our upper and lower bounds are equal. However, our numerical experiments suggest that the gap between the bounds remains small across the tested parameter range; this provides strong evidence that the tight characterization available in the two-attack-type case may hold more generally. These results highlight that the NDWL framework uncovers a fundamental and tractable structure underlying multi-attack-type defense problems, and suggest that there may be important lessons to be learned by analyzing networked resource allocation in more general settings.


\bibliographystyle{ieeetr}

\end{document}